\documentclass[aps,groupedaddress,superscriptaddress,onecolumn]{revtex4-2}
\usepackage{dsfont}
\usepackage[usenames,dvipsnames]{xcolor}
\usepackage[colorlinks,citecolor=Blue,linkcolor=Red, urlcolor=Blue]{hyperref}
\hypersetup{breaklinks=true}
\usepackage[normalem]{ulem}

\usepackage{ifthen}

\newcommand\curRevision{1}
\newcommand{\currentrevision}[1]{
    \def\curRevision{#1}
}

\newcommand{\revdtext}[2][1]{\ifthenelse{\equal{#1}{\curRevision}}{\textcolor{blue}{#2}}{#2}}

\usepackage{graphicx}
\usepackage{subcaption}
\usepackage{float}

\usepackage{tikz}
\usepackage{pgfplots}
\pgfplotsset{compat=1.17}

\usetikzlibrary{positioning,angles,quotes}
\usetikzlibrary{decorations.pathmorphing,arrows.meta}
\usepackage{bbm}
\usepackage{amsmath}
\usepackage{amssymb}
\usepackage{amsthm}
\usepackage{mathtools}

\DeclarePairedDelimiter{\tuples}{(}{)}

\DeclarePairedDelimiter{\brackets}{\{}{\}}

\DeclarePairedDelimiterX\pSet[1]{\{}{\}}{%
  
  #1
}

\ifdefined\skiptheoremsetup\else
    \newtheorem{theorem}{Theorem}

    \newtheorem{lemma}[theorem]{Lemma}

    \newtheorem{definition}{Definition}
    
    \newtheorem{problem}[definition]{Problem}
    
\fi

\newcommand{\C}{\mathbb{C}}

\newcommand{\R}{\mathbb{R}}

\newcommand{\calH}{\mathcal{H}}

\usepackage{xparse}

\newcommand{\bigO}{\mathcal{O}}

\DeclarePairedDelimiterX\pBrackets[1]{[}{]}{%

#1
}

\newcommand{\Prob}[2][]{
 \def\tmp{#1}%
   \ifx\tmp{}
     \operatorname{Pr}\pBrackets*{#2}
   \else
     \operatorname{Pr}_{#1}\pBrackets*{#2}
   \fi
}

\newcommand{\Expect}[2][]{
 \def\tmp{#1}%
   \ifx\tmp{}
     \operatorname{\mathbb{E}}\pBrackets*{#2}
   \else
     \operatorname{\mathbb{E}}_{#1}\pBrackets*{#2}
   \fi
}

\usepackage{physics}
\usepackage{braket}

\usetikzlibrary{quantikz}
\def\qdots{\ \vdots\ }

\renewcommand{\ket}[1]{| #1 \rangle}
\renewcommand{\bra}[1]{\langle #1 |}
\renewcommand{\braket}[2]{\langle #1 | #2 \rangle}
\renewcommand{\proj}[1]{\ket{#1}\bra{#1}}
\renewcommand{\ketbra}[2]{\ket{#1}\bra{#2}}

\newcommand{\diag}{\operatorname{diag}}

\newcommand{\id}{\mathbbm{1}}

\newcommand{\sL}[2]{{\mathfrak{s l}(2, \mathbb{C})}}

\currentrevision{2}

\begin{document}

\title{Robust black-box quantum-state preparation via quantum signal processing}

\author{Lorenzo Laneve}
\email{lorenzo.laneve@usi.ch}
\affiliation{Faculty of Informatics — Universit\`a della Svizzera Italiana, 6900 Lugano, Switzerland}

\begin{abstract}
    Black-box quantum-state preparation is a variant of quantum-state preparation where we want to construct an \revdtext{$n$-qubit} state $\ket{\psi_c} \propto \sum_x c(x) \ket{x}$ with the amplitudes $c(x)$ given as a (quantum) oracle. This variant is particularly useful when the quantum state has a short and simple classical description. We use recent techniques, namely quantum signal processing (QSP) and quantum singular value transform (QSVT), to construct a new algorithm that prepares $\ket{\psi_c}$ without the need to carry out coherent arithmetic. We then compare our result with current state-of-the-art algorithms, showing that a QSVT-based approach achieves comparable results.
\end{abstract}
\maketitle

\section{Introduction}
Quantum signal processing (QSP) is a novel technique for the design of quantum algorithms~\cite{lowMethodologyResonantEquiangular2016}. As a introductory example, consider a unitary $W$: whenever we apply $W$ twice, the resulting operation is $W^2$, regardless of what $W$ is. In other words, this construction applies the polynomial $P(x) = x^2$ to $W$. A natural question arises: which polynomials $P(x)$ can we apply to $W$? Surprisingly it turns out that, with a simple ansatz, we can apply any polynomial satisfying some mild constraints, namely that the polynomial has to be bounded by $1$ in absolute value (natural constraint, as otherwise $P(W)$ cannot be unitary), and of definite parity. The latter constraint can be lifted easily, as one can implement even and odd parts separately, and sum them up using linear combination of unitaries (LCU)~\cite{berryHamiltonianSimulationNearly2015, childsHamiltonianSimulationUsing2012, childsQuantumAlgorithmSystems2017}. This yields a technique called \emph{quantum eigenvalue transform}, which was extensively used to tackle the Hamiltonian simulation problem with surprising (and nearly optimal) complexity~\cite{lowHamiltonianSimulationUniform2017, lowHamiltonianSimulationQubitization2019}, with a more recent construction requiring only a single copy of the initial state~\cite{martynEfficientFullyCoherentQuantum2022}. This idea was further developed by Gilyén et al.~\cite{gilyenQuantumSingularValue2019a}, where the polynomial is applied not on the eigenvalues of the unitary, but rather on the \emph{singular values} of a matrix embedded in the unitary, namely on its top-left block, thus not even requiring this block to be squared. This new technique, called \emph{quantum singular value transform} (QSVT), gives a surprising unification and re-formalization of a wide spectrum of already-known quantum algorithms~\cite{martynGrandUnificationQuantum2021}, from Grover's search~\cite{groverFastQuantumMechanical1996a, groverQuantumMechanicsHelps1997} and amplitude amplification~\cite{brassardQuantumAmplitudeAmplification2002, berryExponentialImprovementPrecision2014a, berrySimulatingHamiltonianDynamics2015} to Shor's factoring~\cite{shorPolynomialTimeAlgorithmsPrime1997}, from quantum phase estimation~\cite{kitaevQuantumMeasurementsAbelian1995, nielsenQuantumComputationQuantum2010b} to the HHL algorithm for solving quantum linear systems~\cite{harrowQuantumAlgorithmSolving2009}.

Quantum-state preparation is a central problem in quantum computation: given complex numbers $c_1, \ldots, c_{2^n}$, we want to construct a circuit that transforms the state $\ket{0}^{\otimes n}$ into the state $\ket{\psi_c} = \sum_{x} c_x \ket{x}$, essentially `initializing' our quantum register for further computation. This has applications, for example, in machine learning~\cite{lloydQuantumPrincipalComponent2014, kerenidisQuantumRecommendationSystems2016} and Hamiltonian simulation~\cite{lowHamiltonianSimulationQubitization2019}, but even techniques such as the LCU itself needs to prepare particular quantum states in order to achieve non-trivial linear combinations~\cite{berryHamiltonianSimulationNearly2015, childsHamiltonianSimulationUsing2012, childsQuantumAlgorithmSystems2017}. Many constructions were devised to prepare an arbitrary state~\cite{barencoElementaryGatesQuantum1995,knillApproximationQuantumCircuits1995,longEfficientSchemeInitializing2001,itenQuantumCircuitsIsometries2016,pleschQuantumstatePreparationUniversal2011}, requiring no ancilla qubits, but exponential depth. In particular, Sun et al.~\cite{sunAsymptoticallyOptimalCircuit2023} found a circuit with depth $\bigO(2^n/n)$ ($n$ being the number of qubits), which matches the lower bound. If we allow ancillary qubits, we obtain depths as low as $\bigO(n)$, although it requires an exponential number of ancillae~\cite{araujoDivideandconquerAlgorithmQuantum2021,sunAsymptoticallyOptimalCircuit2023, rosenthalQueryDepthUpper2022}. Moreover, Zhang et al.~\cite{zhangQuantumStatePreparation2022} improved the complexity under the assumption of sparse states.

Using QSP, a similar problem called \emph{ground-state preparation} has been tackled~\cite{linNearoptimalGroundState2020a,dongGroundStatePreparationEnergy2022}, where one prepares the ground state of a given Hamiltonian. Moreover, QSVT techniques allow to easily prepare Gibbs states on a quantum computer~\cite{chi-fangQuantumThermalState2023}.

In this work we consider the \emph{black-box} quantum-state preparation problem, where the amplitudes $\{ c_x \}_x$ are not given as a list, but as algorithm $c(x)$, with which we can construct a quantum oracle. This idea is nicely applicable if we consider states whose amplitudes are computable (e.g., the purification of a Gibbs state, or a some probability distribution with analytical expression). This problem was originally tackled by Grover as an extension of the search algorithm~\cite{groverSynthesisQuantumSuperpositions2000}, where a controlled rotation was used to carry out amplitude transduction (i.e., transform information to amplitudes): more precisely, this approach uses a rotation by an angle $\theta = \arcsin(x/N)$, which requires an high-accuracy coherent computation of the arcsin function. Sanders et al.~\cite{sandersBlackBoxQuantumState2019} developed a different approach avoiding the need of coherent arithmetic and dramatically reducing the number of gates necessary for an actual implementation: instead of using a rotation, they use a simple comparison algorithm and an additional $\log \epsilon$-qubit register to obtain the amplitudes, provided they are integer multiples of $\epsilon$. McArdle et al.~\cite{mcardleQuantumStatePreparation2022} also use QSVT to prepare a quantum state in a similar setting, where the oracle $c(x)$ is constructed not as a reversible circuit, but as QSVT polynomial, allowing to prepare a state with only $4$ ancilla qubits, but requiring $c(x)$ to have a polynomial approximation that is easy to achieve via quantum signal processing. 

We show that, if $c(x)$ is computable in time $\bigO(T(n))$, then the $N = 2^n$-dimensional quantum state $\ket{\psi_c} \propto \sum_x c(x) \ket{x}$ (with normalization factor) can be prepared within error $\epsilon$ in time $\bigO(\frac{1}{\sqrt{\gamma}} T(n) \log(1/\epsilon))$ and $2 + \lceil \log_2(6/\epsilon \gamma) \rceil$ additional qubits, where
$$\gamma = \frac{1}{N} \sum_x |c(x)|^2 \in (0, 1]$$
is the average squared oracle value (note: $\sqrt{N \gamma}$ is the normalization factor), essentially matching Grover's complexity. This gives a polynomial-time algorithm for the preparation of a large class of quantum states, namely the ones for which the quantum circuit for $c(x)$ is computable in polynomial time and $\sqrt{\gamma}$ is an inverse polynomial in $n$. Indeed, it is worth noting that the number of ancilla qubits needed does not depend directly on $n$.

In Section~\ref{sec:quantum-signal-processing} we give a brief overview of QSP and QSVT, with the necessary elements that we are going to use in the rest of the work. In Section~\ref{sec:unitary-logarithm} we show how to extract the `logarithm' of a unitary using QSVT, a construction taken from~\cite{lowHamiltonianSimulationUniform2017}. In Section~\ref{sec:quantum-state-preparation} we use an ideal implementation of the logarithm of unitary to prepare a quantum state, and Section~\ref{sec:error-analysis} gives a full error analysis when the unitary logarithm is implemented via quantum signal processing.
\section{Review of quantum signal processing}
\label{sec:quantum-signal-processing}
\noindent In this section we briefly describe the quantum signal processing and quantum singular value transform techniques.

\begin{theorem}[Quantum signal processing with reflections~\cite{lowMethodologyResonantEquiangular2016, gilyenQuantumSingularValue2019a}]
    \label{thm:quantum-signal-processing}
    Given the reflection unitary
    \begin{align*}
        R(x) =
        \begin{bmatrix}
            x & \sqrt{1 - x^2} \\
            \sqrt{1 - x^2} & -x
        \end{bmatrix}
    \end{align*}
    and a $d$-degree polynomial $P(x)$ such that:
    \begin{enumerate}
        \item[(i)] has parity $(d \bmod 2)$;
        \item[(ii)] for any $x \in [-1, 1]$, $|P(x)| \le 1$;
        \item[(iii)] for any $x \in (-\infty, -1] \cup [1, \infty)$, $|P(x)| \ge 1$;
        \item[(iv)] when $d$ is even, for any $x \in \R$, $P(ix) P^*(ix) \ge 1$.
    \end{enumerate} 
    The following is true for some $\Phi = (\phi_1, \ldots, \phi_d) \in \R^d$
    \begin{align*}
        \Pi_{j = 1}^d \left( e^{i \phi_j Z} R(x) \right) =
        \begin{bmatrix}
            P(x) & \cdot \\
            \cdot & \cdot
        \end{bmatrix}
    \end{align*}
\end{theorem}
\noindent This results says that we can construct a unitary containing any polynomial $P(x)$ in the top-left corner, provided it satisfies conditions (i)-(iv). We can remove the above unintuitive constraints:
\begin{theorem}[Real quantum signal processing~\cite{gilyenQuantumSingularValue2019a}]
    \label{thm:real-qsp}
    Given a polynomial $P_R \in \R[x]$ satisfying conditions (i)-(ii) of Theorem~\ref{thm:quantum-signal-processing}, there exists a polynomial $P \in \C[x]$ with real part $P_R$ satisfying (i)-(iv).
\end{theorem}
\noindent Thus, for any real polynomial of definite parity and with absolute value bounded by $1$ in our region of interest, the polynomial completed with a suitable imaginary part can be implemented by quantum signal processing. The idea will be to implement both $P, P^*$ (note that the phases $-\Phi$ generate $P^*$), and then implement $P_R = (P+P^*)/2$ via a linear combination of unitaries. It is important to remark that the coefficients of $P$ as well as the phase factors $\Phi$ are computable in polynomial time and with a numerically stable algorithm~\cite{gilyenQuantumSingularValue2019a, haahProductDecompositionPeriodic2019, chaoFindingAnglesQuantum2020}

\begin{definition}[Block-encoding~\cite{gilyenQuantumSingularValue2019a}]
    \label{def:block-encoding}
    Let $A$ be an $s$-qubit matrix and $U$ a $(s+a)$-qubit unitary. We say that $A$ is $(a, \epsilon)$-block-encoded in $U$ if
    \begin{align*}
        \left\lVert A - (\bra{0}^{\otimes a} \otimes \id) U (\ket{0}^{\otimes a} \otimes \id) \right\rVert \le \epsilon
    \end{align*}
\end{definition}
\noindent This means that, if $\epsilon = 0$, $U$ would be of the form
\begin{align*}
    U =
    \begin{bmatrix}
        A & \cdot \\
        \cdot & \cdot
    \end{bmatrix}
\end{align*}
In general, the top-left block of the matrix is $\epsilon$-close to $A$. Notice that, by unitarity of $U$, $\lVert A \rVert \le 1 + \epsilon$. If we need a matrix with norm $\alpha$ we simply block-encode $A / \alpha$. Gilyén et al.~\cite{gilyenQuantumSingularValue2019a} provide a series of constructions which enable different operations on these block encodings.
\begin{definition}[Singular value transformation~\cite{gilyenQuantumSingularValue2019a}]
    \label{def:singular-value-transform}
    Let $A$ be a matrix with singular value decomposition
    \begin{align*}
        A = \sum_i \sigma_i \ketbra{\Tilde{\psi_i}}{\psi_i}
    \end{align*}
    Given a polynomial $P(x)$ of definite parity, the singular value transformation of $A$ using $P$ is defined as
    \begin{align*}
        P^{(SV)}(A) =
        \begin{cases}
            \sum_i P(\sigma_i) \ketbra{\Tilde{\psi_i}}{\psi_i} & \text{$P$ odd} \\
            \sum_i P(\sigma_i) \ketbra{\psi_i}{\psi_i} & \text{$P$ even} \\
        \end{cases}
    \end{align*}
\end{definition}
\noindent It is important to remark that here we could even take the `singular values' to be negative: by negating the corresponding left singular vector we obtain another singular value decomposition, and one can check that Definition~\ref{def:singular-value-transform} remains consistent for any choice of the signs of the singular values. An important question is the following: given a definite-parity polynomial $P(x)$ and a block-encoded matrix $A$, can we obtain a block-encoding of $P^{(SV)}(A)$? The answer is positive.

\begin{theorem}[Quantum singular value transform~\cite{gilyenQuantumSingularValue2019a, martynGrandUnificationQuantum2021}]
    \label{thm:quantum-svt}
    Let $P(x)$ be a polynomial of degree $d$ satisfying (i)-(iv) of Theorem~\ref{def:block-encoding} and let $\Phi \in \R^d$ be the corresponding phase factors. Moreover, $U$ is a unitary that block encodes $A$ as follows
    \begin{align*}
        A = \Tilde{\Pi} U \Pi
    \end{align*}
    where $\Tilde{\Pi}, \Pi$ are projectors. The following unitary produces a block-encoding of $P^{(SV)}(A)$:
    \begin{align*}
        U_\Phi =
        \begin{cases}
            e^{i \phi_1 (2 \Tilde{\Pi} - \id)} U \prod_{k = 1}^{(d-1)/2} \left( e^{i \phi_{2k} (2 \Pi - \id)} U^\dag e^{i \phi_{2k+1} (2 \Tilde{\Pi} - \id)} U \right) & \text{$P$ odd} \\
            \prod_{k = 1}^{d/2} \left( e^{i \phi_{2k-1} (2 \Pi - \id)} U^\dag e^{i \phi_{2k} (2 \Tilde{\Pi} - \id)} U \right) & \text{$P$ even}
        \end{cases}
    \end{align*}
    i.e., $\Tilde{\Pi} U_{\Phi} \Pi = P^{(SV)}(A)$.
\end{theorem}
\noindent Roughly speaking, the proof shows that, considering the subspaces spanned by the $i$-th singular vectors, $U$ acts as $R(\sigma_i)$, while the rotations acts as $Z$-rotations, and we can apply Theorem~\ref{thm:quantum-signal-processing} on each of these subspaces. This enables to carry out a singular value transform using any polynomial constructible with quantum signal processing. This transformation is also robust, in the sense that, if $A$ is $(a, \epsilon)$-block-encoded in $U$, then we have that $P^{(SV)}(A)$ is $(a+1, 4d\sqrt{\epsilon})$-block-encoded in $U_{\Phi}$~\cite{gilyenQuantumSingularValue2019a}. Most of the time we will focus on the case where $A$ is Hermitian. In this case the singular value and eigenvalue transformations will coincide (remember that here, we can intend singular values to also be negative), i.e., $P^{(SV)}(A) = P(A)$.
\section{The logarithm of a unitary}
\label{sec:unitary-logarithm}

\noindent Keeping in mind what we presented in the last section, consider the following problem.
\begin{problem}[Unitary logarithm]
    \label{def:unitary-logarithm}
    Let $\calH$ be an Hermitian matrix satisfying $\lVert \calH \rVert < 1$, and denote with $U = e^{i \pi \calH}$ the corresponding unitary. Given controlled versions of $U, U^\dag$, implement a block-encoding $C$ of $\calH$, i.e.,
    \begin{align*}
        C =
        \begin{bmatrix}
            \calH & \cdot \\
            \cdot & \cdot
        \end{bmatrix}
    \end{align*}
\end{problem}
\noindent We now show a simple way introduced in~\cite{lowHamiltonianSimulationUniform2017} to solve Problem~\ref{def:unitary-logarithm}. This is used in~\cite{gilyenQuantumSingularValue2019a} to make fractional queries to $U$, i.e., to implement $U^t$ for a non-necessarily integer $t$, namely by extracting the Hamiltonian, multiplying it by a constant using block-encoding arithmetics, and then exponentiating it back with a polynomial that approximates the complex exponential function. Another example in the same work, which we are going to generalize, was done for Gibbs sampling.
The idea is as follows: by doing some simple calculations one can check that
\begin{align} 
    \label{eq:hamiltonian-sine-construction}
    (\bra{0} H \otimes \id) cU^\dag (Y \otimes \id) cU (H \ket{0} \otimes \id) = \sin(\pi \calH)
\end{align}
i.e., this simple circuit (as shown in Figure~\ref{fig:hamiltonian-sine-construction}) is a $(1, 0)$-block-encoding for $\sin(\pi \calH)$, where $H$ is the Hadamard gate, $cU$ is the controlled-$U$ gate, and $Y$ is the Pauli matrix. In order to obtain a block-encoding of $\calH$, we need to invert the sine function.
\begin{figure}
    \centering
    \input{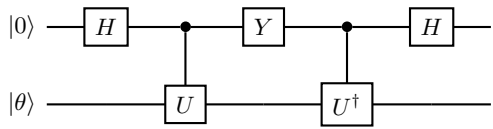}
    \caption{Construction of the sine of an Hamiltonian. One can notice that, if $\ket{\theta}$ is an eigenstate of $U$ with eigenvalue $e^{i\pi \theta}$, at the end we obtain the state $\sin(\pi\theta) \ket{0} - i \cos(\pi\theta) \ket{1}$. In other words, the unitary $C$ representing this circuit perfectly block-encodes $\sin(\pi \calH)$.}
    \label{fig:hamiltonian-sine-construction}
\end{figure}

\begin{theorem}[\cite{gilyenQuantumSingularValue2019a}]
    \label{thm:arcsin-approx}
    An $\epsilon$-polynomial approximation of $f(x) = \frac{1}{\pi} \arcsin(x)$ in the interval $(-1 + \delta, 1 - \delta)$ has degree $d = \bigO\tuples*{\frac{1}{\delta} \log \frac{1}{\epsilon}}$.
\end{theorem}
\begin{proof}
    The Taylor series of $f(x)$ is
    \begin{align*}
        f(x) = \sum_{k = 0}^\infty \frac{1}{\pi} \binom{2 k}{k} \frac{2^{-2k}}{2k+1} x^{2k+1}
    \end{align*}
    and the truncation up to the first $d$ terms is $\epsilon$-close to $f$ in the interval $(-1 + \delta, 1 - \delta)$~\cite[Theorem 68]{gilyenQuantumSingularValue2019a}.
\end{proof}

\noindent Therefore, assuming $\lVert \calH \rVert < 1 - \delta$ we can construct a $(2, \epsilon)$-block-encoding of $\calH$ with only $\bigO(\frac{1}{\delta} \log \frac{1}{\epsilon})$ calls to $cU, cU^\dag$.
\section{Quantum-state preparation}
\label{sec:quantum-state-preparation}

\noindent The quantum-state preparation problem can be stated without loss of generality as follows:
\begin{problem}
    \label{def:quantum-state-prep}
    Let $N = 2^n$. Given amplitudes $c = (c_0, \cdots, c_{N-1}) \in [0,1]^N$, construct the state
    $$\ket{\psi_c} = \sum_x c_x \ket{x}$$
    from the state $\ket{0}^{\otimes n}$ up to error $\epsilon$. More formally, construct a quantum circuit $C$ such that
    \begin{align*}
        \left\lVert C \ket{0}^{\otimes n} - \ket{\psi_c} \right\rVert \le \epsilon
    \end{align*}
\end{problem}
\noindent It will be clear later why we do not need $c_i$ to be complex. In this work, we consider the \emph{black-box} quantum-state preparation problem, where we assume the amplitudes are computed by an algorithm $c(x) \in [0, 1]$, and this algorithm is given as a quantum oracle
\begin{align*}
    \bigO_c \ket{x} \ket{0}^{\otimes m} = \ket{x} \ket{c(x)}
\end{align*}
which computes the $m$ bits after the decimal point. Moreover, we will not need to assume that the $c^2(x)$ are normalized, but we will assume the algorithm has access to the average $\gamma = \frac{1}{N} \sum_{x} c^2(x)$. Thus, in the end the target state will be
\begin{align*}
    \ket{\psi_c} = \frac{1}{\sqrt{N \gamma}} \sum_{x} c(x) \ket{x}
\end{align*}
The idea is quite simple: consider the unitary $U_c$ acting as follows
\begin{align*}
    U_c \ket{x} = e^{i \pi c(x) / 2} \ket{x}
\end{align*}
thus $U_c = \diag(e^{i \pi c(x) / 2})_x = e^{i\pi H_c}$ where
$$H_c = \diag(c(x)/2)_x.$$
This unitary is actually efficiently implementable using only two copies of $\bigO_c, \bigO_c^\dag$, using a standard construction (Figure~\ref{fig:phase-oracle-construction}). This is the reason why we only care for positive real amplitudes, as applying relative phases is always efficiently realizable with a similar transformation. Extracting the logarithm of this unitary using the construction of Section~\ref{sec:unitary-logarithm}, yields a $(2, \epsilon)$-block-encoding of $H_c$ using only $\bigO(\log \frac{1}{\epsilon})$ calls to $\bigO_c$ (notice that $\lVert H_c \rVert \le \frac{1}{2}$). We denote the unitary of this block-encoding with $C$.
\begin{figure}
    \centering
    \input{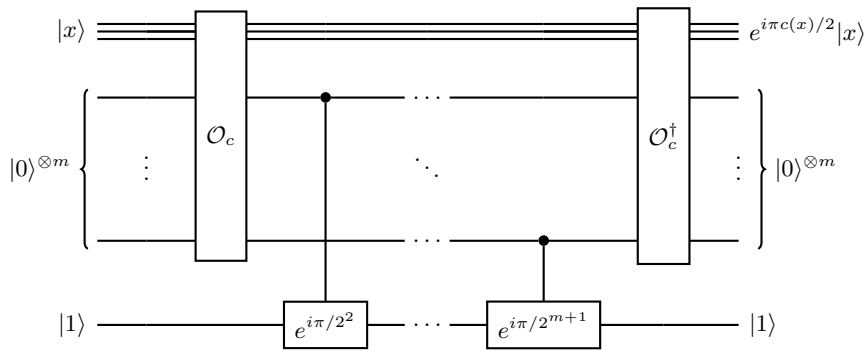}
    \caption{Construction of $U_c$ using two copies of $\bigO_c, \bigO^\dag_c$. The controlled phase rotations give a contribution for each bit of the output, so that the total phase obtained is $\pi c(x)/2$. In order to obtain a $\epsilon$-approximation on $c(x)$ it is sufficient to take $m = \bigO(\log(1/\epsilon))$ ancilla qubits. Note that we do not really need the additional qubit, which is added only for ease of exposition.}
    \label{fig:phase-oracle-construction}
\end{figure}

We assume for ease of exposition that the $\arcsin$ approximation is perfect, and we will postpone the error analysis to the next section. If we now apply this operator to the equal superposition $\ket{+} := \ket{+}^{\otimes n}$ we obtain
\begin{align}
    \label{eq:state-prep-block-encoding}
    C \ket{00} \ket{+} = \ket{00} H_c \ket{+} + \ket{\Phi}
\end{align}
where $\ket{00}$ is the initial state of the two control qubits, and $\ket{\Phi}$ is the garbage state we obtain if we fail, i.e., we pick the wrong block of the block encoding and the two control qubits return $\neq 00$ (which means $\braket{00}{\Phi} = 0$). The state associated with the $\ket{00}$ component is
\begin{align*}
    H_c \ket{+} & = \frac{1}{\sqrt{N}} \sum_x H_c \ket{x} = \frac{1}{2\sqrt{N}} \sum_x c(x) \ket{x} = \frac{\sqrt{\gamma}}{2} \ket{\psi_c}
\end{align*}
Thus, by replacing this in Eq.~(\ref{eq:state-prep-block-encoding}) we obtain
\begin{align*}
    C \ket{00} \ket{+} = \frac{\sqrt{\gamma}}{2} \ket{00} \ket{\psi_c} + \ket{\Phi}
\end{align*}
This means that, if we measure the control qubits, we post-select the correct block and get our state with probability $\gamma/4$. In order to amplify the success probability we employ a fixed-point amplitude amplification procedure, proving its correctness with the aid of Theorem~\ref{thm:quantum-svt}.

\begin{lemma}[Fixed-point amplitude amplification]
    \label{thm:non-unitary-amplitude-amplification}
    Let $\ket{\Psi} = \ket{00} \ket{\psi_a} = \ket{00} S \ket{0^n}$ be our initial state and $\ket{w} = \ket{00} \ket{\psi_c}$ is the target state. Given the unitary $C$ acting as
    \begin{align*}
        C \ket{\Psi} = \sigma \ket{w} + \ket{\Phi}
    \end{align*}
    where $\braket{00}{\Phi} = 0$, i.e., $\ket{\Phi}$ is our `garbage' state, given by the other block of the encoding. It is possible to obtain $\ket{w}$ with probability $1 - \delta$ using $\bigO(\frac{1}{\sigma} \log \frac{1}{\delta})$ copies of $C$ and $S$. 
\end{lemma}
\noindent 
\iffalse
The crucial difference between this construction and the already-known oblivious amplitude amplification procedures~\cite{berryExponentialImprovementPrecision2014a, berrySimulatingHamiltonianDynamics2015, yanFixedpointObliviousQuantum2022} is that here we do not need the desired transformation of our initial state to be unitary (or close to a unitary). However, we also require many copies of the unitary $S$ that constructs the initial state, i.e., the initial state is fixed. \fi In our application, the initial state is $\ket{+}$, so $S = H^{\otimes n}$ is the $n$-fold Hadamard gate, which is easy to construct.
\begin{proof}
    Consider the two projectors
    \begin{align*}
        \Tilde{\Pi} & = \proj{00} \otimes \id \\
        \Pi & = \proj{\Psi} = \proj{00} \otimes S \proj{0^n} S^\dag
    \end{align*}
    Then
    \begin{align*}
        \Tilde{\Pi} C \Pi = \sigma \ketbra{w}{\Psi}
    \end{align*}
    i.e., $C$ block encodes a rank-one matrix with a singular value $\sigma$. All we need to do is to design a  $P(x)$ that satisfies $|P(\sigma)| \ge 1 - \delta/2$. In this way, we apply Theorem~\ref{thm:quantum-svt} to transform this singular value, so our success probability will be $\ge (1 - \delta/2)^2 \ge 1 - \delta$. A polynomial approximation $P$ to the sign function can achieve this~\cite[Corollary~6]{lowHamiltonianSimulationUniform2017}, and has degree $\bigO(\frac{1}{\Delta} \log \frac{1}{\delta})$:
    \begin{align*}
        |P(x)| \ge 1 - \delta/2 \,\,\, \text{for $|x| \ge \Delta$}
    \end{align*}
    and we plug $\Delta = \sigma$ (see Figure~\ref{fig:sign-function-approx}).
\end{proof}
\noindent Therefore, as $\sigma = \sqrt{\gamma}/2$, by applying the construction of Theorem~\ref{thm:non-unitary-amplitude-amplification}, we obtain
\begin{lemma}
    \label{thm:perfect-quantum-state-perp}
    Starting from the state $\ket{0}^{\otimes n}$, we can construct the state $\ket{\psi_c}$ with probability $1 - \delta$ using
    $$\bigO\left(\frac{1}{\sqrt{\gamma}} \log \frac{1}{\delta}\right)$$
    copies of $C$.
\end{lemma}

\begin{figure}
    \centering
    \input{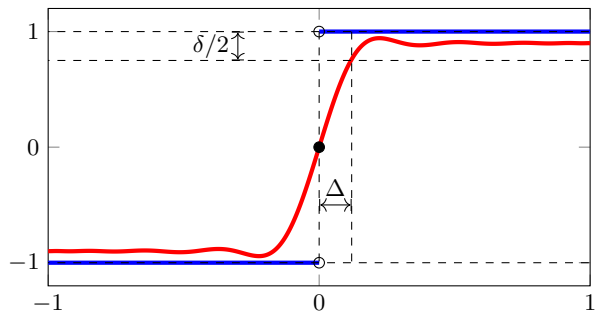}
    \caption{Approximation of the sign function using an odd polynomial. We use a polynomial of degree $\bigO(\frac{1}{\Delta} \log \frac{1}{\delta})$ to obtain a $\delta/2$-approximation of the sign function when $|x| \ge \Delta$. By increasing the degree we can increase both the accuracy of the approximation and the range. In amplitude amplification settings, the singular value we want to amplify usually sits close to $0$, and we want to transform it as close as possible to $1$.}
    \label{fig:sign-function-approx}
\end{figure}
\section{Error analysis}
\label{sec:error-analysis}
\noindent In the previous section we considered $C$ being a perfect block-encoding of $H_c = \arcsin(\sin H_c)$ and we proved that, under this assumption, our algorithm delivers exactly $\ket{\psi_c}$ with probability $1 - \delta$. We now replace the unitary $C$ with some unitary $\Tilde{C}$ such that
\begin{align*}
    \Tilde{C} \ket{00} \ket{+} = \ket{00} \Tilde{H}_c \ket{+} + \ket{\Tilde{\Phi}}
\end{align*}
where
\begin{align}
    \label{eq:error-hamiltonian-spectral-norm}
    \lVert \Tilde{H}_c - H_c \rVert \le \epsilon.
\end{align}
Notice that such $\Tilde{C}$ can be actually implemented using $\bigO(\log \frac{1}{\epsilon})$ calls to $\bigO_c$, as per Theorem~\ref{thm:arcsin-approx} (for now we assume that $c(x)$ has finite precision and is perfectly computed by $\bigO_c$, we deal with the case where $c(x)$ has to be approximated by an $m$-bit truncation towards the end of the section). Thus, the eigenvalues $\frac{1}{2}\Tilde{c}(x)$ of $\Tilde{H}_c$ satisfy $|\frac{1}{2} \Tilde{c}(x) - \frac{1}{2} c(x)| \le \epsilon$, and $\Tilde{H}_c \ket{+}$ will give us the sub-normalized state
\begin{align*}
    \Tilde{H}_c \ket{+} & = \frac{1}{2 \sqrt{N}} \sum_x \Tilde{c}(x) \ket{x} =: \frac{\sqrt{\Tilde{\gamma}}}{2} \ket{\Tilde{\psi}_c}
\end{align*}
where $\ket{\Tilde{\psi}_c}$ and $\Tilde{\gamma}$ are defined analogously as $\ket{\psi_c}$ and $\gamma$. In particular, the former is the final state returned by our algorithm, after the amplification procedure. Thus our task is to bound $\lVert \ket{\Tilde{\psi}_c} - \ket{\psi_c}\rVert$. The first observation is that, by using Eq.~(\ref{eq:error-hamiltonian-spectral-norm})
\begin{align}
    \label{eq:asymmetric-distance-bound}
    \left\lVert \frac{\sqrt{\Tilde{\gamma}}}{2} \ket{\Tilde{\psi}_c} - \frac{\sqrt{\gamma}}{2} \ket{\psi_c} \right\rVert & = \left\lVert (\Tilde{H}_c - H_c) \ket{+} \right\rVert \le \left\rVert \Tilde{H}_c - H_c \right\lVert \le \epsilon
\end{align}
A second observation is that $\gamma, \Tilde{\gamma}$ are close
\begin{align*}
    |\Tilde{\gamma} - \gamma| & \le \frac{1}{N} \sum_{x} |\Tilde{c}^2(x) - c^2(x)| = \frac{1}{N} \sum_{x} |\Tilde{c}(x) - c(x)| \cdot |\Tilde{c}(x) + c(x)| \le 2 \epsilon
\end{align*}
Thus, the distance between the two square roots is
\begin{align*}
    |\sqrt{\gamma} - \sqrt{\Tilde{\gamma}}| & = \frac{|\Tilde{\gamma} - \gamma|}{|\sqrt{\gamma} + \sqrt{\Tilde{\gamma}}|} \le \frac{\epsilon}{2 \min\brackets{\sqrt{\gamma}, \sqrt{\Tilde{\gamma}}}}
\end{align*}
If we assume that $\epsilon \le \frac{\gamma}{4}$, then $\Tilde{\gamma} \ge \gamma/2$ and the bound will become
\begin{align}
    \label{eq:average-distance-sqrt-bound}
    |\sqrt{\gamma} - \sqrt{\Tilde{\gamma}}| & \le \frac{\epsilon}{\sqrt{2} \sqrt{\gamma} }
\end{align}
Thus, the total error will be
\begin{align*}
    \left\lVert \ket{\Tilde{\psi}_c} - \ket{\psi_c} \right\rVert & = \frac{2}{\sqrt{\gamma}} \left\lVert \frac{\sqrt{\gamma}}{2}\ket{\Tilde{\psi}_c} - \frac{\sqrt{\gamma}}{2}\ket{\psi_c} \right\rVert \\
    & = \frac{2}{\sqrt{\gamma}} \left\lVert \frac{\sqrt{\gamma} - \sqrt{\Tilde{\gamma}}}{2}\ket{\Tilde{\psi}_c} + \frac{\sqrt{\Tilde{\gamma}}}{2}\ket{\Tilde{\psi}_c} - \frac{\sqrt{\gamma}}{2}\ket{\psi_c} \right\rVert \\
    & \le \frac{\left| \sqrt{\gamma} - \sqrt{\Tilde{\gamma}} \right|}{\sqrt{\gamma}} + \frac{2}{\sqrt{\gamma}} \left\lVert \frac{\sqrt{\Tilde{\gamma}}}{2}\ket{\Tilde{\psi}_c} - \frac{\sqrt{\gamma}}{2}\ket{\psi_c} \right\rVert \\
    & \le \frac{\epsilon}{\gamma \sqrt{2}} + \frac{2\epsilon}{\sqrt{\gamma}} \le \frac{3\epsilon}{\gamma}
\end{align*}
where we used Eqs.~(\ref{eq:asymmetric-distance-bound})-(\ref{eq:average-distance-sqrt-bound}) to bound the two terms at the end. The amplitude amplification procedure of Theorem~\ref{thm:non-unitary-amplitude-amplification} has be run with $\sigma = \sqrt{\Tilde{\gamma}}/2 = \Omega(\sqrt{\gamma})$, so $\bigO(\frac{1}{\sqrt{\gamma}} \log \frac{1}{\delta})$ copies of $C$ still suffice.

We proved that we get an error bound of $3\epsilon/\gamma$ with probability $1 - \delta$ using the procedure described in Section~\ref{sec:quantum-state-preparation}, and it requires a total of
$$\bigO\tuples*{\frac{1}{\sqrt{\gamma}} \log \tuples*{\frac{1}{\delta}} \log\tuples*{\frac{1}{\epsilon}}}$$
calls to the oracle for $c(x)$. Since we want error bound $\epsilon$, we simply replace $\epsilon \leftarrow \frac{\epsilon \gamma}{3}$ in the above argument. This change of variable does not alter the asymptotic complexity, as we assumed $\epsilon \le \gamma/4$.
\begin{theorem}[Robust black-box quantum-state preparation]
    Let $c : [2^n] \rightarrow [0,1]$ be a function implemented by a quantum circuit
    \begin{align*}
        \bigO_c \ket{x} \ket{0}^{\otimes m} = \ket{x} \ket{c(x)}
    \end{align*}
    It is possible to construct the normalized state $\ket{\psi_c} = \frac{1}{\sqrt{N\gamma}} \sum_x c(x) \ket{x}$ using $\bigO(\frac{1}{\sqrt{\gamma}} \log \frac{1}{\delta} \log \frac{1}{\epsilon})$ copies of $\bigO_c$ and single- and two-qubit gates.
\end{theorem}
\noindent We highlight here that the error bound by Eq.~(\ref{eq:error-hamiltonian-spectral-norm}) can be used not only to bound the approximation error of the quantum signal processing polynomial for the arcsin. Noise of other nature can arise, for example the imperfection of the gates in Eq.~(\ref{eq:hamiltonian-sine-construction}), or the fact that $c(x)$ can be an arbitrary real number and the oracle only computes a $m$-bit representation of it. However, all these noises will be treated by the analysis of this section with little effort. Importantly, we deal with the case where $c(x)$ has too high (or even infinite) precision, and our oracle has to approximate it with an $m$-bit truncation $c'(x)$. 
In this case, the block encoding we are going to extract is $\Tilde{H}_{c'}$, i.e.\ the tilde on $H$ represents the error given by the polynomial approximation, as in $\Tilde{H}_c$. Thus we can bound the error
\begin{align*}
    \left\lVert \Tilde{H}_{c'} - H_c \right\rVert \le \left\lVert \Tilde{H}_{c'} - H_{c'} \right\rVert + \left\lVert H_{c'} - H_c \right\rVert.
\end{align*}
The first term is bounded by $\epsilon$, which is the accuracy of the polynomial approximation, and the second term only depends on the accuracy of the oracle, which can be also bounded by $\epsilon$ by taking $m = \lceil\log_2(1/\epsilon)\rceil$. Thus the distance between the extracted Hamiltonian $\Tilde{H}_{\Tilde{c}}$ and the ideal one $H_c$ is bounded by $2\epsilon$. We can replace $\epsilon \leftarrow \epsilon/2$ without changing the complexity and obtain the bound of Eq.~(\ref{eq:error-hamiltonian-spectral-norm}), from which the rest of the analysis is identical. Thus, at the end of the analysis, the number of additional ancilla qubits will be $2 + \lceil \log_2(6/\epsilon \gamma) \rceil$ (the two qubits come from the QSVT ansätze, while the factor $6/\epsilon\gamma$ is due to the two replacements of $\epsilon$ we have done throughout our analysis).
\section{Comparison with other approaches}
\noindent Here we show a comparison with other quantum algorithms solving black-box quantum-state preparation in terms of depth, query and qubit complexity. Grover's approach in~\cite{groverSynthesisQuantumSuperpositions2000} works as follows: we do the following operations
\begin{align*}
    \ket{k} \ket{0} \ket{0}_S \ket{0}_T & \stackrel{\bigO_c}{\rightarrow} \ket{k} \ket{c(k)} \ket{0}_S \ket{0}_T \\
    & \rightarrow \ket{k} \ket{c(k)} \ket{\theta_k}_S \ket{0}_T \\
    & \rightarrow \ket{k} \ket{c(k)} \ket{\theta_k}_S (\sin \theta_k \ket{0}_T + \cos \theta_k \ket{1}_T) \\
    & \simeq c(k) \ket{k} \ket{c(k)} \ket{\theta_k}_S \ket{0}_T + \ket{junk}
\end{align*}
where $T$ is an additional qubit used for amplitude transduction, $S$ is register containing sufficiently many qubits where the approximation $\theta_k$ of the arcsine is saved. After uncomputing the two middle registers this achieves the desired state with $\bigO(1/\sqrt{\gamma} \log(1/\epsilon) \log(1/\delta))$ queries and $n + 2\lceil\log(1/\gamma\epsilon)\rceil + \bigO(1)$ qubits after an amplitude amplification over $\ket{0}_T$. The problem with this method is the approximation of the arcsine, which has been estimated to take over 11000 Toffoli gates~\cite{häner2018optimizing}. Sanders et al.~\cite{sandersBlackBoxQuantumState2019} replace the arcsin implementation with a simple comparison procedure:
\begin{align*}
    C \ket{x} \ket{y} \ket{0}_T =
    \begin{cases}
        \ket{x} \ket{y} \ket{0}_T & x < y \\
        \ket{x} \ket{y} \ket{1}_T & x \ge y
    \end{cases}
\end{align*}

\noindent By considering $c(k)$ to be a $m$-bit number (as opposed to an $m$-bit expansion of a number $\in [0, 1]$ in our analysis), we obtain:
\begin{align*}
    \ket{k} \ket{0} \ket{0}_S \ket{0}_T & \rightarrow \ket{k} \ket{c(k)} \ket{0}_S \ket{0}_T \\ &
    \stackrel{H^{\otimes m}_S}\rightarrow \ket{k} \ket{c(k)} \left(\sum_{x = 0}^{2^m - 1} \ket{x}_S\right) \ket{0}_T \\
    & \stackrel{C}\rightarrow \ket{k} \ket{c(k)} \left(\sum_{x = 0}^{c(k) - 1} \ket{x}_S \ket{0}_T + \sum_{x = c(k)}^{2^m - 1} \ket{x}_S \ket{1}_T \right) \\
    & \stackrel{H^{\otimes m}_S}\rightarrow \frac{c(k)}{2^m} \ket{k}\ket{c(k)}\ket{0}_S\ket{0}_T + \ket{junk}
\end{align*}
Thus we obtain the same effect after uncomputing $c(k)$ and applying amplitude amplification over $\ket{0}_S \ket{0}_T$, requiring the same asymptotic query and qubit complexity, but using significantly less gates at each round of amplitude amplification.

Our approach is more similar to Grover's, as we also work by implementing the arcsine. However, the construction of Eq.~(\ref{eq:hamiltonian-sine-construction}) essentially transduces the amplitudes $\sin c(x)$ directly, and the arcsine is applied afterwards, in the form of a QSVT polynomial, rather than a coherent implementation. While the gate and query complexities match the previous approaches, we need $n + \lceil \log(1/\gamma\epsilon) \rceil + \bigO(1)$ qubits. Notice that the $\gamma$ factor in the logarithm, both present in our approach and Grover's, is needed to compensate the numerical error of the arcsine approximation, which is amplified by amplitude amplification. This is not a problem in the algorithm in~\cite{sandersBlackBoxQuantumState2019}, as their method does not introduce any sort of error.

\begin{figure}
    \centering
    \setlength{\tabcolsep}{10pt}
    \renewcommand{\arraystretch}{2}
    \begin{tabular}{r c c c c c} 
        \hline
        & Queries to $\bigO_c$ & Qubits & CA \\ 
        \hline
        Grover~\cite{groverSynthesisQuantumSuperpositions2000} & $\bigO(\frac{1}{\sqrt{\gamma}} \log \epsilon \log \delta)$ & $n + 2 \log \frac{1}{\gamma\epsilon} + \bigO(1)$ & Yes \\
        Sanders et al.~\cite{sandersBlackBoxQuantumState2019} & $\bigO(\frac{1}{\sqrt{\gamma}} \log \epsilon \log \delta)$ & $n + 2 \log \frac{1}{\epsilon} + \bigO(1)$ & No \\
        This work & $\bigO(\frac{1}{\sqrt{\gamma}} \log \epsilon \log \delta)$ & $n + \log \frac{1}{\gamma\epsilon} + \bigO(1)$ & No \\
        \hline
        
\end{tabular}
    \caption{Comparison of qubit complexity black-box quantum-state preparation algorithms. CA = Coherent arithmetic.}
    \label{fig:comparison-table}
\end{figure}
\section{Discussion}
\noindent In this paper we devised a new algorithm for quantum state preparation, where the amplitudes are given as quantum oracles $c(x)$, like in Grover's search~\cite{groverFastQuantumMechanical1996a, groverQuantumMechanicsHelps1997}. Speaking of Grover's search, one can see that the unstructured search problem can be seen as a special case of black-box quantum state preparation, where $c(x)$ is $1$ for exactly one value $x_0$, and $0$ for the others. In this case $\gamma = N$ and our algorithm takes $\bigO(\sqrt{N} \log \frac{1}{\delta} \log \frac{1}{\epsilon})$ queries to prepare $\ket{x_0}$ as in Grover's algorithm ($\delta$ and $\epsilon$ can be unified into a single probability of failure upon measurement of the state). The dependence on $\gamma$ tells us that the construction of $\ket{\psi_c}$ is tightly related to how `uniform' the distribution is: if we want to construct something close to the $\ket{+}^{\otimes n}$ state, the oracle $c(x)$ will have values close to $1$, and $\gamma \simeq 1$. On the other hand, if we have a few spikes in our state (the limit case is the one of unstructured search above, with only one spike), then the average $\gamma \ll 1$, and this will impact significantly on the complexity. Notice that $\gamma$ need not be smaller than $1/N$: if $\max c(x) < 1$, then we can divide every value of $c(x)$ by this number and work with $\max c(x) = 1$.

It is important to remark, however, that we bound the query complexity for the oracle $c$ in our analysis, but it can be a challenging problem to construct a polynomial-time quantum circuit for $c(x)$ if the amplitudes $c_1, \ldots, c_N$ are truly random. Perhaps quantum algorithmic information theory can be of help in understanding how much we can `compress' a given list of amplitudes into a polynomial-time algorithm~\cite{vitanyiQuantumKolmogorovComplexity2001}. For the easier-to-handle case where the state to prepare has a nice analytical expression, however, the time complexity of the construction is only given by the normalization factor $\gamma$. We highlight here again the fact that the above algorithm only constructs states with positive real amplitudes, but complementing this algorithm with a second oracle $\phi(x)$ constructed as in Figure~\ref{fig:phase-oracle-construction} allows one to obtain amplitudes of the form $c(x) e^{i\pi \phi(x)}$.

\section*{Acknowledgements}
\noindent I would like to thank William Schober, Stefan Wolf and Charles B\'edard for insightful feedback and discussions. This work was supported by the Swiss National Science Foundation (SNF), grant No. \texttt{200020\_182452}.

\bibliography{refs}

\end{document}